\documentclass[%
 reprint,
superscriptaddress,
 amsmath,amssymb,
 aps,
]{revtex4-2}

\usepackage{graphicx}% Include figure files
\usepackage{dcolumn}% Align table columns on decimal point
\usepackage{bm}% bold math
\usepackage{amsmath}
\usepackage{amsthm}
\usepackage{amssymb}
\usepackage{amssymb}
\usepackage{mathtools}
\usepackage{braket}
\usepackage{dsfont}
\usepackage{hyperref}
\usepackage{tikz}
\usetikzlibrary{arrows,positioning} 
\usetikzlibrary{quantikz}

\usepackage{mathtools} % for 'bsmallmatrix' environment
\newenvironment{brsm}{% % short for 'bracketed small matrix'
   \begin{smallmatrix} }{%
  \end{smallmatrix} }
\newcommand{\id}{\mathrm{id}}
\newcommand{\comment}[1]{}

\newcommand{\CNOT}{\rm CNOT}
\newcommand{\SWAP}{\rm SWAP}
\newtheorem{theorem}{Theorem}%[section]
\newtheorem{remark}[theorem]{Remark}%[section]
\newtheorem{corollary}[theorem]{Corollary}%[section]
\newtheorem{proposition}[theorem]{Proposition}%[section]
\DeclareMathOperator{\im}{\mathring \imath}

\begin{document}

\preprint{APS/123-QED}

\title{Transpiling Quantum Circuits using the Pentagon Equation}% Force line breaks with \\

\author{Christos Aravanis}
\email{c.aravanis@sheffield.ac.uk}
\affiliation{
 The University of Sheffield International College, 3 Solly Street, Sheffield, S1 4DE, U.K.}

\author{Georgios Korpas}
\email{georgios.korpas@hsbc.com} 
\affiliation{%
HSBC Lab, 
Digital Partnerships \& Innovation, 8 Canada Square, London, E14 5HQ, U.K.}
%\affiliation{%
%Department of Computer Science, 
%Czech Technical University in Prague, Karlovo nam. 13, Prague 2, Czech Republic}


\author{Jakub Marecek}
\email{jakub.marecek@fel.cvut.cz}
\affiliation{%
Department of Computer Science, 
Czech Technical University in Prague, Karlovo nam. 13, Prague 2, Czech Republic}

\date{\today}% It is always \today, today,

\begin{abstract}
We consider the application of the pentagon equation in the context of quantum circuit compression. We show that if solutions to the pentagon equation are found, one can transpile a circuit involving  
non-Heisenberg-type interactions
to a circuit involving only
Heisenberg-type interactions
while, in parallel, reducing the depth of a circuit. In this context, we consider a model of non-local two-qubit operations of Zhang \emph{et. al.} (the $A$ gate), and show that for certain parameters it is a solution of the pentagon equation.
\end{abstract}

\maketitle
%\tableofcontents

\section{Introduction}
Quantum technologies and quantum computing in particular are advancing in an unprecedented pace. As a result, the capabilities of quantum computers are constantly increasing wherein the first sings of ``quantum supremacy" or ``quantum advantage" \cite{arute2019quantum,zhong2020quantum} hint that it might not take as long as originally thought to adopt them in a large scale. However, the current generation of quantum computers, often termed the Noisy Intermediate-Scale Quantum (NISQ) computers \cite{Preskill2018quantumcomputingin}, suffer from various limitations which make them quite hard to exploit for useful applications across various domains such as finance, logistics, chemistry or material science.
A key limitation is known as the (effective) depth of a quantum circuit. 
In the quantum circuit model, one applies a sequence of unitary transformations $U$ on quantum states \cite{nielsen_chuang_2010}. The number of subsequent applications of unitary transformations on (up to) two qubits at a time, is known as the depth of the circuit. 
The depth of the circuit is limited by the ratio of coherence time and gate time,  
as well as the fidelity of the two-qubit gates. 
The limit on the depth of a quantum circuit limits what algorithms can be implemented, and the quality of the output, despite the flourishing of various error mitigation techniques. While Fault Tolerant Error Corrected (FTEC) quantum computers will be more resilient to these limitations, one may envision that even there, it will be preferable to reduce the number of gates. 
Circuit transpiling techniques, which can reduce circuit depth, will be essential in both the NISQ and FTEC eras of quantum computing.
The transpiling of quantum circuits starting from a potentially large gate set to a particular hardware-specific ``native'' gate set, can be quite challenging, esp. in architectures with limited connectivity \cite{lin2014paqcs}. Limited connectivity is another key limitation of the current quantum computers and several challenges must be overcome to fully exploit these powerful machines. 
Notice that, due to the limited connectivity, two-qubit gates often cannot be readily applied and the states of the corresponding qubits must be swapped with those of their neighbors until the states reside on qubits where a two qubit gate is supported. SWAP gates are, however, often expensive. For instance in CNOT-based native gate sets, they are often implemented using three CNOT gates. As a result, reducing the number of SWAP gates is desirable, but computing the minimum number of SWAP gates required in a given circuit is an NP-Hard problem \cite{botea2018complexity}.
Most techniques for transpiling quantum circuits are multi-pass heuristics, rather than exact \cite[cf.]{grosse2009exact,nannicini2021optimal,madden2022best,madden2022first}.
First, an initial set of transformations is used to translate the quantum circuit  to a ``native'' gate set.
Second, a heuristic mapping of logical qubits to physical qubits is suggested.  
Third, standard ``circuit compression'' transformations \cite{1012662,1219016} are applied in a Knuth-Bendix fashion. 
Ref. % 1219016 has 500+ citations and is something of a classic.
\cite{1012662} introduced the first five templates and many more followed \cite[e.g.]{shende2005synthesis,maslov2004reversible,1219016,maslov2008quantum,Gulania2021}.
We refer to \cite{saeedi2013synthesis,kusyk2021survey} for more in-depth surveys
and to \cite{zulehner2018efficient,tan2020optimality} for recent experimental comparisons.

%Current transpilers, such as the one in IBM Qiskit\footnote{https://qiskit.org/documentation/apidoc/transpiler.html} are rather sophisticated.
 
%Recent research has considered various alternatives. 
%In Ref. \cite{PhysRevA.105.032420}, in the context of Hamiltonian simulation, an algorithm for compression of the Trotterization procedure to a single block of gates is given. For certain classes of Hamiltonians, this algorithm yields a fixed depth time evolution. In Ref. \cite{CompressionZX}, in the context of Quantum Error Correction (QEC), circuit compression for the 3D topological code was proposed using ZX-calculus \cite{CompressionZX}. 

\begin{figure}[!htb]
    \centering
    \includegraphics[scale=0.166]{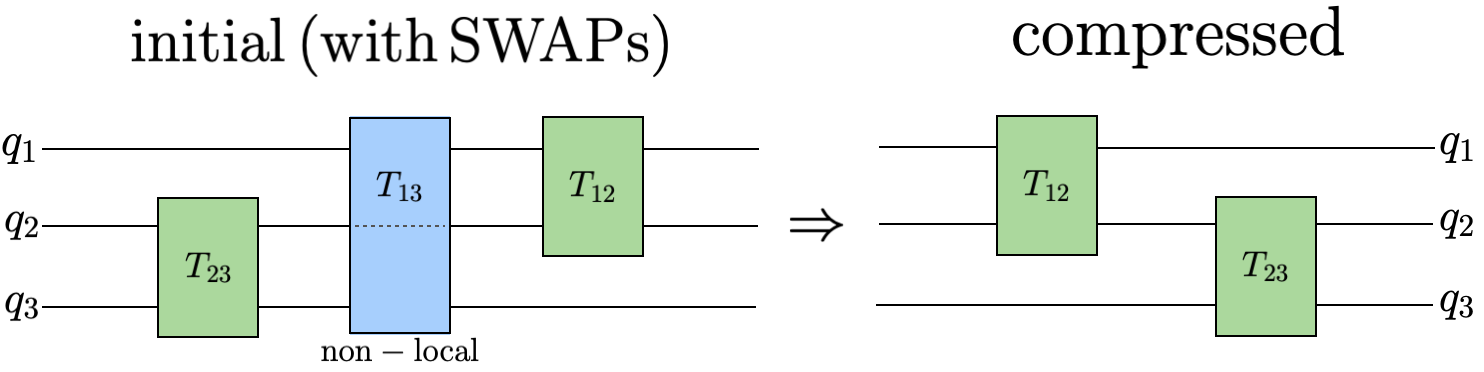}
    \caption{The schematic of a circuit compression for gates $T$ that satisfy the pentagon equation. In the LHS there is a non-local 2-qubit gate that involves SWAP gates.}
    \label{fig:fancy}
\end{figure}

Our work is inspired by Ref. \cite{Gulania2021}. There, the authors approached the problem of quantum circuit compression, in the framework of the Hamiltonian evolution of the 1D Heisenberg model, using the Yang-Baxter Equation (YBE) \cite{Bethe1931,Baxter1985,FADDEEV1995,JIMBO1989,Chari1995ov}. (See \cite{isaev2022lectures} for a modern survey of YBE.) Specifically, they map the YBE to parametrized quantum circuits composed of 2-qubit gates and find the necessary conditions on these parameters such that the gates permuted in a certain way can be compressed by adding the parameter values of two sequential gates on the same two qubits. In this context, solutions of the YBE have the potential to significantly reduce the depth of a local circuit.

%%%%%%% GIORGOS 27/9 %%%%%%%%%%%%%%
%%%%%%% GIORGOS 27/9 %%%%%%%%%%%%%%
%%%%%%% GIORGOS 27/9 %%%%%%%%%%%%%%

In this note, we elaborate upon the approach of Ref. \cite{Gulania2021} by compressing the depth of quantum circuits using lesser known tools from integrability theory, specifically the so-called ``pentagon equation". However, our method yields a surprise on the potential applications on quantum circuits, a form of duality. %In particular, our work can be viewed in two ways. 

First, different to Ref. \cite{Gulania2021}, we show the necessary conditions required to compress the depth of quantum circuits that contain 2-qubit non-local interactions, that is non-nearest-neighbor long-range interactions, that usually are implemented by running first SWAP gates (see Figs. \ref{fig:fancy} and \ref{fig:nonlocal}). This can achieve circuit compression under certain conditions. The schematic of the compression is shown in Fig. \ref{fig:fancy}.
In a more general context, our method obtains a dual picture where circuits with non-local interactions can be expressed as circuits with local ones as long or vice-versa.  Our approach depends on these 2-qubit gates that implement this ``duality" to be solutions of the so-called pentagon equation \cite{StreetFusion}, Eq. \eqref{def:pen}, and we focus on the implementation using the evolution operator of the 1D Heisenberg model \cite{Baxter} as well as the so-called $A$ gate \cite{Vala_shortPaper}. 
%%%%%%% GIORGOS 27/9 %%%%%%%%%%%%%%
%%%%%%% GIORGOS 27/9 %%%%%%%%%%%%%%
%%%%%%% GIORGOS 27/9 %%%%%%%%%%%%%%

%First, we show the necessary conditions required to compress the depth of quantum circuits, different to Ref. \cite{Gulania2021}, at the ``cost'' of introducing non-local interactions, that is non-nearest-neighbor long-range interactions . Second, our method obtains a dual picture where circuits with non-local interactions can be expressed as circuits with local ones at the cost of introducing extra gates. Our approach depends on these gates that implement this ``duality" to be solutions of the pentagon equation \cite{StreetFusion} and we focus on the implementation using the evolution operator of the 1D Heisenberg model \cite{Baxter} as well the so-called $A$ gate \cite{Vala_shortPaper}. 

This note is organized as follows: 
In Section \ref{alg-setup}, we collect all the algebraic preliminaries we need for the rest of the paper and fix some notation. In Section \ref{Heisenberg model}, we collect all the necessary material for understanding the 1D Heisenberg model. In Section \ref{Pentagon equation}, we introduce the pentagon equation and  conclude with a digression to Yang-Baxter equation. In Section \ref{Results} we present our results: we discuss initially the form of duality we obtain from the application of the pentagon equation. Then, we explain how to transpile quantum circuits and conclude with the necessary conditions for the evolution operator of the 1D Heisenberg mode and the $A$-gate to satisfy the pentagon equation.
%In Section \ref{App} we show how the pentagon equation can be used to compress a quantum circuit.

\begin{figure}[!htb]
    \centering
    \includegraphics[scale=0.25]{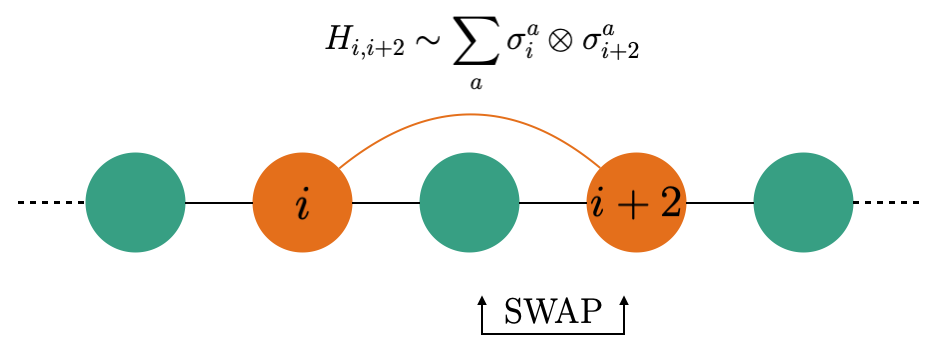}
    \caption{To implement the non-local interaction $H_{i,i+2}$ on the $i$-th and $i+2$-th qubits, SWAP gates must be implemented.}
    \label{fig:nonlocal}
\end{figure}

%and heuristics need to be developed that come close to the ideal solution while scaling favorably with the number of qubits.

\section{Algebraic setup} \label{alg-setup}
Throughout the paper, we will denote a finite dimensional vector space over the complex numbers by  $V$ with the usual tensor product of vector spaces $\otimes$. The identity map from $V$ to $V$ will be denoted $\id_{V}$ and the twist map $\tau_{V,V}\colon V \otimes V \to V\otimes V$ reads $x\otimes y \mapsto y\otimes x$. A linear map $f\colon U \to V$, between two vector spaces $U$ and $V$, has an associated matrix with respect a basis of $U$ and $V$. By abusing notation, we will denote the linear map and its associated matrix by the same letter.%  Let $W$ be another finite dimensional vector space over the complex number and $g\colon V \to W$ be a linear map. The composite map $g\circ f$ has matrix obtained by$B_{g}A_{f}$.For a linear map $h: U'\to V'$ the tensor product $f\otimes h\colon U\otimes U' \to V\otimes V'$ has associated matrix $A_{f}\otimes C_{h}$, the Kronecker product of the associated matrices $A_{f}$, $C_{h}$.

The Pauli matrices $\sigma_a$, $a \in \{ x,y,z \}$ which are the generators of the Lie algebra $\mathfrak{su}(2)$ of the Lie group $SU(2)$ read
\begin{align*}
    \sigma_x = \begin{pmatrix} 0 & 1 \\ 1 & 0   \end{pmatrix}, \quad \sigma_y = \begin{pmatrix} 0 & -\im \\ \im & 0   \end{pmatrix}, \quad \sigma_z = \begin{pmatrix} 1 & 0 \\ 0 & -1   \end{pmatrix},
\end{align*}
where $\im$ is the imaginary unit. It is worth mentioning here that $\sigma_{x}^{2}=\sigma_{y}^{2}=\sigma_{z}^{2}=\mathds{1}$. % \footnote{CA: checked with Mathematica}. 
Based on this fact, one computes \[\exp(\im \theta \sigma_{a})=\cos(\theta)\mathds{1}+\im\sin(\theta)\sigma_{a},\] for $a=x,y,z$ and $\theta \in \mathbb{R}$. Precisely, the matrix exponentials of $\sigma_x$, $\sigma_y $ and $\sigma_z$ give rise to the rotation operator matrices $R_{x}(\theta)$, $R_{y}(\theta)$ and $R_{z}(\theta)$ respectively which read 
    \begin{equation*}
    \begin{aligned}
     R_{x}(\theta) &=  \begin{pmatrix} \cos(\theta/2) & -\im\sin(\theta/2) \\ -\im\sin(\theta/2) & \cos(\theta/2)   \end{pmatrix}, \\[1.2em] 
     R_{y}(\theta) &=  \begin{pmatrix} \cos(\theta/2) & \sin(\theta/2) \\ \sin(\theta/2) & \cos(\theta/2)   \end{pmatrix}, \\[1.2em]
     R_{z}(\theta) &= \begin{pmatrix} e^{-\im\theta/2} & 0 \\ 0 & e^{\im\theta/2}   \end{pmatrix}.
    \end{aligned}
\end{equation*}

\subsection*{Standard Gates}

A quantum circuit is a series of unitary transformations, called gates, which act on qubits. Gates which will appear in this note are: the Hadamard gate $H$ and the phase shift gate $S$ which  are applied to single qubits and read respectively
\begin{align*}
   H =  \frac{1}{\sqrt{2}}\begin{pmatrix} 1 & 1 \\ 1 & -1   \end{pmatrix}, \qquad S = \begin{pmatrix} 1 & 0 \\ 0 & \im   \end{pmatrix}.
\end{align*} The controlled NOT gate, denoted by CNOT, and the SWAP which are applied to 2-qubits and read respectively
\begin{align*}
    \CNOT = \begin{pmatrix}
1 & 0 & 0 & 0\\
0 & 1 &  0 & 0 \\
0 & 0 & 0 & 1 \\
0 & 0 & 1 & 0
\end{pmatrix}, \qquad \SWAP=\begin{pmatrix}
1 & 0 & 0 & 0\\
0 & 0 &  1 & 0 \\
0 & 1 & 0 & 0 \\
0 & 0 & 0 & 1
\end{pmatrix}
\end{align*} 
Note that the inverse of the SWAP matrix is equal to itself, symbolically holds $\SWAP^{-1} = \SWAP$.

\subsection*{The $A$ and $B$ gates}

While there are many standard gates, it would be convenient to reason about all non-local two-qubit operations at the same time, while considering as few parameters as possible. To do so, Zhang \emph{et al.} \cite{Vala_longPaper, Vala_shortPaper} introduced the three-parameter $A$ gate, where each choice of the three parameters corresponds to 
a local equivalence class of two-qubit gates. 
%, the $A$ and $B$ gates, . 
In particular, the gate  $A$  is defined as the product of the unitary matrices  $XX(c_1) \coloneqq e^{\im \frac{c_{1}}{2} \sigma_x \otimes \sigma_x}$, $ YY(c_2) \coloneqq e^{\im\frac{c_{2}}{2} \sigma_y \otimes \sigma_y}$ and $ZZ(c_3) \coloneqq e^{\im\frac{c_{3}}{2} \sigma_z \otimes \sigma_z}$ where
\begin{align*}
    XX(c_1) &= \begin{pmatrix} \cos \tfrac{c_1}{2} & 0 & 0 & \im \sin \tfrac{c_1}{2}  \\
                            0 &  \cos \tfrac{c_1}{2} & \im \sin \tfrac{c_1}{2} &0 \\
                            0 & \im \sin \tfrac{c_1}{2} & \cos \tfrac{c_1}{2} & 0 \\
                            \im \sin\tfrac{c_1}{2} &0 &0 & \cos \tfrac{c_1}{2} 
               \end{pmatrix}, \\[1.2em]
    YY(c_2) &= \begin{pmatrix} \cos \tfrac{c_2}{2} & 0 & 0 & -\im \sin \tfrac{c_2}{2}  \\
                            0 &  \cos \tfrac{c_2}{2} & \im \sin \tfrac{c_2}{2} &0 \\
                            0 & \im \sin \tfrac{c_2}{2} & \cos \tfrac{c_2}{2} & 0 \\
                            -\im \sin\tfrac{c_2}{2} &0 &0 & \cos \tfrac{c_2}{2} \
               \end{pmatrix},\\[1.2em]
    ZZ(c_3) &= \begin{pmatrix} 
 e^{\tfrac{\im c_3}{2}} & 0 & 0 & 0 \\
 0 & e^{-\tfrac{\im c_3}{2}} & 0 & 0 \\
 0 & 0 & e^{-\tfrac{\im c_3}{2}} & 0 \\
 0 & 0 & 0 & e^{\tfrac{\im c_3}{2}} \\
\end{pmatrix}.
\end{align*}

The analytic form of the $A$ gate is displayed in Fig.~\ref{fig:agate}.
See also \cite{Vala_longPaper, Vala_shortPaper} for the elegant geometry thereof.

\begin{figure*}[t!]
%\begin{widetext}
\begin{equation*}\label{Agate}
A = \begin{pmatrix}
  e^{\frac{\im c_3}{2}} \cos\left(\frac{1}{2}\left(c_1-c_2\right)  \right) & 0 & 0 &  \im e^{\frac{\im c_3}{2}} \sin\left(\frac{1}{2}\left(c_1-c_2\right)  \right)  \\
 0 &  e^{-\frac{\im c_3}{2}} \cos\left(\frac{1}{2}\left(c_1+c_2\right)  \right) & \im e^{\frac{-\im c_3}{2}} \sin\left(\frac{1}{2}\left(c_1+c_2\right)  \right) & 0 \\
 0 & \im e^{\frac{-\im c_3}{2}} \sin\left(\frac{1}{2}\left(c_1+c_2\right)  \right)  & e^{-\frac{\im c_3}{2}} \cos\left(\frac{1}{2}\left(c_1+c_2\right)  \right) & 0 \\
\im e^{\frac{\im c_3}{2}} \sin\left(\frac{1}{2}\left(c_1-c_2\right)  \right)  & 0 & 0 &    e^{\frac{\im c_3}{2}} \cos\left(\frac{1}{2}\left(c_1-c_2\right)  \right)
\end{pmatrix},
\end{equation*} 
%\end{widetext}
\caption{An analytic form of the $A$ gate, where $c_{1}, c_{2}, c_{3}$ are integer coefficients.}
\label{fig:agate}
\end{figure*}

Associated with the $A$ gate, but  not as important as the $A$ gate for our considerations, is the $B$ gate which reads $B = e^{\frac{\pi \im}{4}  \sigma_x^1 \otimes \sigma_x^2} \cdot e^{\frac{\pi \im}{8} \sigma_y^1 \otimes \sigma_y^2}$. The $B$ gate is related to the $A$ gate as $A \sim B U B$, for some $U\in U(2)\times U(2)$ (see \cite[page 2]{Vala_longPaper}). 

%thus we direct the reader to the original paper to see more details about it. 

\section{The 1D Heisenberg model} \label{Heisenberg model}
The Heisenberg model \cite{Sutherland2004} is a spin model of ferromagnetism on a lattice where the coupling energy $J$ between nearest neighbor lattice sites is positive and parallel alignment of local spins is favorable. The variables in the Heisenberg model are subject to a continuous internal symmetry which once broken to $\mathbb{Z}_2$ it yields the Ising model which is quite commonly used in the context of Variational Quantum Algorithms (VQAs) \cite{cerezo2021variational}. Both the Heisenberg and the Ising models are of particular interest in the theory of quantum integrability precisely due to their integrable nature; they satisfy certain equations for which one can find analytic solutions at any value of the coupling energy $J$ and the thermodynamic limit $N \to \infty$, where $N$ is the number of lattice sites or spins.

Our work is motivated by Ref. \cite{Gulania2021} which we summarize briefly. In this context,
we are interested especially for the 1D Heisenberg model with $N=2$. This is a case of particular interest since it is the only scenario where the components of the model's Hamiltonian commute. The Heisenberg Hamiltonian is defined as
\begin{equation}
\hat{H} = - \sum_{a} J_{a}(\sigma^{a}_{1}\otimes \sigma^{a}_{2})
\end{equation} where $a \in \{x,y,z\}$, $\sigma^{a}$ is the Pauli matrix at the direction of $a$, $J_{\alpha}$ is the coupling constant or interaction strength. The evolution operator of the Shr\"odinger equation is $e^{\im \hat{H}t/\hbar}$ is displayed in Fig.~\ref{fig:Heisenberg}.

\begin{figure*}[t!]
%\begin{widetext}
\begin{equation*}\label{def:evOp}
\begin{aligned}
 e^{\im\hat{H}t/\hbar}=\prod_{a} e^{\im J_{a}t(\sigma^{a}_{1}\otimes \sigma^{a}_{2})/\hbar} =\begin{pmatrix}
e^{\im\theta_{z}}\cos(\theta_{x}-\theta_{y}) & 0 & 0 & \im e^{\im\theta_{z}}\sin(\theta_{x}-\theta_{y})\\
0 & e^{-\im\theta_{z}}\cos(\theta_{x}+\theta_{y}) &  \im e^{-\im\theta_{z}}\sin(\theta_{x}+\theta_{y}) & 0 \\
0&\im e^{-\im\theta_{z}}\sin(\theta_{x}+\theta_{y})  &e^{-\im\theta_{z}}\cos(\theta_{x}+\theta_{y}) &0 \\
\im e^{-\im\theta_{z}}\sin(\theta_{x}-\theta_{y}) &0&0&e^{\im\theta_{z}}\cos(\theta_{x}-\theta_{y}),
\end{pmatrix}
\end{aligned}
\end{equation*}
%\end{widetext}
%\begin{widetext}
\caption{An analytic form of the  evolution operator of the Shr\"odinger equation.}
\label{fig:Heisenberg}
\end{figure*}

Notice that $\gamma = \theta_{x}-\theta_{y}$ and $\theta_{x}+\theta_{y}= \theta_{x}+\theta_{y}$ in Fig. \ref{fig:Heisenberg} is based on a calculation of the individual components:
\begin{equation*}
    \begin{aligned}
        e^{\im J_{x}t(\sigma^{x}_{1}\otimes \sigma^{x}_{2})/\hbar} &=\left(\begin{brsm}
        \cos(\theta_{x}) & 0 & 0 & \sin(\theta_{x})\\
        0 & \cos(\theta_{x}) &  \im \sin(\theta_{x}) & 0 \\
        0 & \im \sin(\theta_{x})  & \cos(\theta_{x}) & 0 \\
        \im \sin(\theta_{x}) & 0 & 0 \cos(\theta_{x}) \
    \end{brsm}\right) \\
        e^{\im J_{y}t(\sigma^{y}_{1}\otimes \sigma^{y}_{2})/\hbar} &=\left(\begin{brsm}
        \cos(\theta_{x}) & 0 & 0 & -\im \sin(\theta_{y})\\
        0 & \cos(\theta_{y}) &  \im \sin(\theta_{y}) & 0 \\
        0 & \im \sin(\theta_{y})  & \cos(\theta_{y}) & 0 \\
        -\im \sin(\theta_{y}) & 0 & 0 & \cos(\theta_{y})
        \end{brsm}\right) \\
        e^{\im J_{z}t(\sigma^{z}_{1}\otimes \sigma^{z}_{2})/\hbar} &=\left(\begin{brsm}
        e^{\im \theta_{z}} & 0 & 0 & 0 \\
        0 & e^{-\im\theta_{z}} & 0 & 0 \\
        0 & 0  & e^{-\im\theta_{z}} & 0 \\
        0 & 0 & 0 & e^{\im\theta_{z}}
\end{brsm}\right).
    \end{aligned}
\end{equation*}

\begin{remark}\label{remark1}
One obtains the matrix  $e^{\im\hat{H}t/\hbar}$ from the $A$ gate by setting $c_{1}=2(\theta_{x}-\theta_{y})$, $c_{2}=2(\theta_{x}+\theta_{y})$ and $c_{3}=2\theta_{z}$ in matrix $A$.
\end{remark}The evolution operator of the 1D Heisenberg model as a quantum circuit appears in Fig.~\ref{fig:QuantumCirc}.

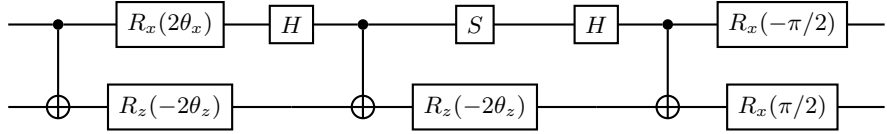
\begin{figure*}[t!]
    \begin{quantikz}
   \qw & \ctrl{1}  & \gate{R_{x}(2\theta_{x})}   &  \gate{H} & \ctrl{1} & \gate{S} & \gate{H} &\ctrl{1} & \gate{R_{x}(-\pi/2)} & \qw  \\
   \qw & \targ{}   & \gate{R_{z}(-2\theta_{z})}  & \qw       & \targ{}  &  \gate{R_{z}(-2\theta_{z})} & \qw & \targ{} & \gate{R_{x}(\pi/2)} & \qw   
    \end{quantikz}
    \caption{An optimal quantum circuit of $e^{\im\hat{H}t/\hbar}$ of 1D Heisenberg model, see  \cite{Gulania2021}. }
    \label{fig:QuantumCirc}
\end{figure*}

\section{Pentagon equation} \label{Pentagon equation}
The pentagon equation is an equation, which belongs in a infinite family of equations called polygon equations, and is similar to the Yang-Baxter equation. It appears in many branches of mathematics, such as representation theory and topological field theory \cite{Kirillov1990}, conformal field theory \cite{Belavin:1984vu}, Teichm\"uller theory \cite{fock1997dual}, Hopf algebras, quantum dilogarithm \cite{faddeev1994quantum,volkov2003noncommutative}, and operator theory, to name a few. See  \cite{Dimakis} and the rich bibliography. The pentagon equation has found applications in the context of topological quantum computing previously, see \cite{kitaev2006anyons} and \cite[Ch. 4]{Kasirajan2021}, as well as in the context of symmetries of non-local matrix product operators \cite{lootens2021matrix}. 

In this section, we define the pentagon equation. Then, we make  a digression on Yang-Baxter equation, pointing on what the two equations differ and on what the two equations look similar.

Let $V$ be a finite dimensional and let  $T\colon V\otimes V \to V \otimes V$ be a linear map. Define the maps \begin{equation*}
    T_{12}, T_{23}, T_{13}\colon V\otimes V\otimes V \to V\otimes V\otimes V
\end{equation*} by the formulae
\begin{equation*}
    \begin{aligned}
      T_{12}  &\coloneqq T \otimes \id_{V}
\\
 T_{23}  &\coloneqq \id_{V} \otimes T \\
 T_{13} &\coloneqq (\id_{V}  \otimes \tau_{V,V})^{-1} \circ (T\otimes \id_{V}) \circ (\id_{V} \otimes\tau_{V,V}).
    \end{aligned}
\end{equation*}
The \emph{pentagon equation} is defined by
\begin{align}\label{def:pen}
 T_{23}T_{12} = T_{12}T_{13}T_{23}.
\end{align} A solution of the pentagon equation is a linear operator $T$ which satisfies Eq. \eqref{def:pen}. In the bibliography, solutions of the pentagon equations are called a \emph{fusion operators}. Finding solutions of the pentagon equation, otherwise fusion operators, is a challenging process and an active area of research. 

\begin{remark}\label{rem:recipe}
One way to produce fusion operators, i.e. solutions of the pentagon equation, is from a bialgebra. A bialgebra $B$ is a vector space with an algebra structure and a compatible coalgebra structure. Denote by $m$ the product of the algebra and by $\theta_{x}+\theta_{y}$ the coproduct of the coalgebra. Then, the composite map $T:=(id \otimes m) \circ (\delta \otimes id)$ is a fusion operator, see \cite[Proposition 1.2]{StreetFusion} for more details. See also \cite{kashaev1995heisenberg} for a more physical oriented approach.
\end{remark}

 %\textcolor{red}{Precisely, we have the following theorem.}

One of our interests is to make use of the pentagon equation and compress quantum circuit, that means to reduce the number of gates. For this purpose, we depict the pentagon equation \eqref{def:pen} as a quantum circuit as shown in Fig. \ref{fig:fusion}. %(with the use of $Quantikz$ Package, see \cite{Kay2019} for more details).
\begin{figure*}[t!]
%\begin{widetext}
\begin{center}
\begin{quantikz}
 & \gate[2]{T}  &  \qw           & \qw & \\
&              & \gate[2]{T} & \qw &  \\
  & \qw          &             & \qw  &
 \end{quantikz} \quad $=$ \quad \begin{quantikz}
& \qw          & \qw                      &  \gate[2]{T}    &  \qw                         & \gate[2]{T}               & \qw & \ \\
 & \gate[2]{T}  & \gate[swap]{}            &                 &\gate[swap]{}                 &                      & \qw & \\
 &              &                          & \qw             &                              &   \qw                & \qw & 
\end{quantikz}

 \comment{

\begin{quantikz}
\lstick{$a$} & \qw          & \qw                      &  \gate[2]{T}    &  \qw                         & \gate[2]{T}               & \qw & \ \\
& \gate[2]{T}  & \gate[2]{\textsc{SWAP}}  &                 &\gate[2]{\textsc{SWAP}}  &       & \qw &  \\
 &              &                          & \qw             &                              &   \qw                & \qw & 
\end{quantikz}
}
\end{center}
%\end{widetext}
\caption{The circuit description of the pentagon equation \eqref{def:pen}. Twisting the quantum wires is  implemented by enforcing non-local interactions. Even within the landscape of NISQ devices, e.g. in certain neutral atom architectures, this is possible by directly applying appropriate microwave pulses.} %  The fusion operator
\label{fig:fusion}
\end{figure*}
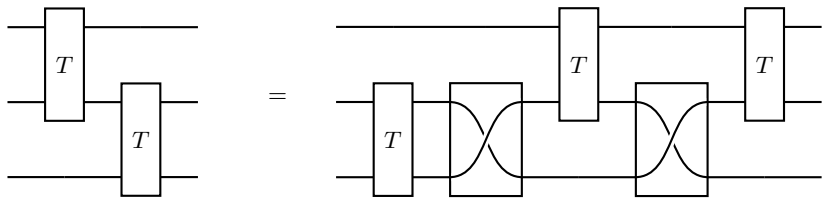

\subsubsection*{Digression: Yang-Baxter equation}

In this section, we will highlight the connection between the pentagon equation, which is the main topic of this paper with the rather famous Yang Baxter equation. The Yang-Baxter equation  has significant applications in the theory of 2D integrable systems and the theory of quantum groups, see \cite{FADDEEV1995,JIMBO1989,Chari1995ov,isaev2022lectures} and also has been used in the context of quantum computing, see \cite{Kauffman2004}.

For a finite dimensional vector space $V$ and a linear map $R\colon V\otimes V \to V\otimes V$ the Yang-Baxter equation is defined by  \begin{equation}\label{YBE}
    (R \otimes \id)(\id \otimes R)(R \otimes \id) = (\id \otimes R)(R \otimes \id)(\id \otimes R) 
\end{equation}or equivalently by
 \begin{align} \label{QYBE}
    R_{12}R_{23}R_{12} = R_{23}R_{12}R_{23}.
\end{align} where the subscripts indicate which tensor factors are being utilized. The quantum circuit implementation of the YBE \eqref{QYBE} is depicted in Fig. \ref{fig:ybe1}.

Among the solutions of the Yang-Baxter equation are the identity map $\id \colon V\otimes V\to V\otimes V$ and the transposition map $\tau\colon V\otimes V\to V\otimes V$ which reads $\tau(v_{1}\otimes v_{2}):= v_{2}\otimes v_{1}$. In Ref. \cite{Kauffman2004}, the authors showed that the Bell basis change
matrix consists of a solution of the Yang-Baxter equation and, moreover, it is a universal gate for quantum computing, among others. Other solutions may be obtained as universal $R$ elements of quansi triangular Hopf algebras \cite{ICMDrinfeld}.

\begin{figure}
    \centering
   \begin{tikzpicture}
\node[scale=0.75] 
{
\begin{quantikz}
 & \gate[2]{R}   &  \qw             & \gate[2]{R}    & \qw     & \\
 &              & \gate[2]{R}      &                & \qw     &  \\
 & \qw          &                  & \qw             &\qw     &
 \end{quantikz} \quad $=$ \begin{quantikz}
 & \qw                 &  \gate[2]{R}      &  \qw            & \qw     &\\
&   \gate[2]{R}       &                   & \gate[2]{R}     & \qw     &  \\
&                     &   \qw             & \qw             &\qw     &
 \end{quantikz}
 };
\end{tikzpicture}
    \caption{The circuit description of the Yang-Baxter Equation~\eqref{QYBE}.}
    \label{fig:ybe1}
\end{figure}
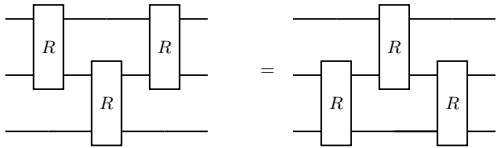

\begin{theorem}(Foklore) \label{thm: QYBE}
Let $R\colon V \otimes V \to V \otimes V$ be a linear map. Then $R$ is a solution of the Yang-Baxter equation \ref{QYBE} if and only if $R':=\tau_{V,V}\circ R$ satisfies the equation 
 \begin{equation} \label{QYBE2}
    R_{12}R_{13}R_{23} = R_{23}R_{13}R_{12} 
\end{equation}
\end{theorem} The circuit description of the equation \eqref{QYBE2} is shown in Fig.~\ref{fig:pent}.

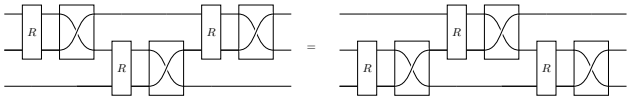
\begin{figure}
    \centering
    \begin{tikzpicture}
\node[scale=0.48] 
{
\begin{quantikz}
 &    \gate[2]{R}    &         \gate[swap]{} & \qw    & \qw           & \gate[2]{R}        &   \gate[swap]{}      & \qw    \\
 &    \qw             &       \qw        &  \gate[2]{R}       & \gate[swap]{} &    \qw        &          \qw           & \qw  \\
 &     \qw            &     \qw           &\qw      & \qw           &         \qw     &  \qw    & \qw      
\end{quantikz} \quad $=$ \quad

\begin{quantikz}
& \qw               &   \qw            & \gate[2]{R}    &         \gate[swap]{}  &\qw   &   \qw             & \qw \\
 &    \gate[2]{R}    &         \gate[swap]{}   &  \qw            &   \qw         & \gate[2]{R}    &         \gate[swap]{}  & \qw    \\
  &     \qw         &  \qw             &   \qw           & \qw           & \qw   &  \qw             & \qw
\end{quantikz} 
\label{fig:3co}
};
\end{tikzpicture}
    \caption{The circuit description of the Eq~\eqref{QYBE2}.}
    \label{fig:pent}
\end{figure}

\begin{remark}
   The pentagon equation \eqref{def:pen} is Eq.~\eqref{QYBE2} with the term $R_{13}$ omitted on the right-hand side.
\end{remark}

There is an analogous result to Theorem~\ref{thm: QYBE} for the pentagon equation.

\begin{theorem}(Street \cite{StreetFusion}) \label{thm: Street}
Let $T\colon V \otimes V \to V \otimes V$ be a linear map. Then, $T$ is a fusion operator if and only if $T':=\tau_{V,V}\circ T$ satisfies the 3-cocycle condition 
\begin{equation}\label{def:3-cocycle}
   \begin{aligned}
 (T' \otimes {\id}_{V}) \circ ({\id}_{V} \otimes \tau_{V,V}) \circ (T' \otimes {\id}_{V}) \\
 = ({\id}_{A} \otimes T')\circ (T' \otimes {\id}_{V}) \circ ({\id}_{V} \otimes T')
\end{aligned}  
\end{equation}

\end{theorem} The circuit description of the 3-cocycle is shown in Fig~\ref{fig:3co}.

\comment{
\begin{center}
\begin{tikzpicture}
\node[scale=0.7] 
{
\begin{quantikz}
\lstick{$a$} & \gate[2]{T'}        & \qw                      &  \gate[2]{T'}       & \qw  & \lstick{$a$}  \\
\lstick{$b$} &                     & \gate[swap]{}  &                     & \qw &\lstick{$c$}  \\
\lstick{$c$} &     \qw             &                          & \qw                 &  \qw  & \lstick{$b$}         
\end{quantikz} \quad $=$ \quad \begin{quantikz}
\lstick{$a$} &   \qw            & \gate[2]{T'}      &   \qw             & \qw &\lstick{$c$} \\
\lstick{$b$} &   \gate[2]{T'}   &                   &   \gate[2]{T'}    & \qw  &\lstick{$b$}   \\
\lstick{$c$} &                  &   \qw             &    \qw             & \qw &\lstick{$a$}
\end{quantikz} 
\label{fig:3co}
};
\end{tikzpicture}
\end{center}

%The notation might be confusing here since in the r.h.s. we see 
The 3-cocycle circuit is easier to understand if we consider the wire twists in full detail:}

\begin{figure}
    \centering
    \begin{tikzpicture}
\node[scale=0.54] 
{\begin{quantikz}
 &    \gate[2]{T}    &         \gate[swap]{}    & \qw           & \gate[2]{T}        &   \gate[swap]{}      & \qw    \\
 &    \qw             &       \qw               & \gate[swap]{} &    \qw        &          \qw           & \qw  \\
 &     \qw            &     \qw                 & \qw           &         \qw     &  \qw    & \qw      
\end{quantikz} \quad $=$ \quad 
\begin{quantikz}
& \qw               &   \qw            & \gate[2]{T}    &         \gate[swap]{}  &\qw   &   \qw             & \qw \\
 &    \gate[2]{T}    &         \gate[swap]{}   &  \qw            &   \qw         & \gate[2]{T}    &         \gate[swap]{}  & \qw    \\
  &     \qw         &  \qw             &   \qw           & \qw           & \qw   &  \qw             & \qw
\end{quantikz} 

\label{fig:3co}
};
\end{tikzpicture}
    \caption{The circuit description of the 3-cocycle~\eqref{def:3-cocycle}.}
    \label{fig:3co}
\end{figure}
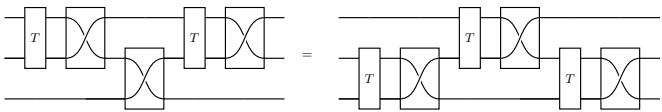

Reading this circuit from right to left is more instructive for reasons that will become apparent in the next section. Naturally, the easiest way to implement this type of circuit, should a fusion operator $T$ exist, is to use a SWAP gate. As a matter of fact, in what follows we assume that quantum wire twist is indeed implemented using SWAP gates. However, as mentioned in Fig. \ref{fig:fusion}, in certain architectures, for example in the case of analog neutral atom approaches, it is possible to directly implement such a twist when indeed one is interested in the next-to-nearest-neighbor interactions, just as shown above.

% To make sense of the above diagram look the commutative diagram below which should be read from left to right.

%By analogy the Pentagon equation is 
%\begin{align}  (T \otimes id)(id \otimes \tau)(T \otimes id) = (id \otimes T)(T \otimes id)(id \otimes T)  \end{align}

%\begin{definition} The linear map $T\colon V \otimes V \to V \otimes V $ is a fusion operator if satisfies the pentagon equation below % A fully expanded form based on the explicit definition of $T_{12}$, $T_{23}$, $T_{13}$ is 
%\begin{align} (T\otimes id_{V})  \circ \big((id_{V} \otimes \tau_{V,V})^{-1} \circ (T\otimes id_{V}) \circ (id_{V} \otimes \tau_{V,V})\big) \circ (id_{V} \otimes T)  =   (id_{V} \otimes T) \circ (T\otimes id_{V})\end{align} \end{definition}

% Solution of the pentagon equation is given by Hopf algebras as follows.
%\textbf{Example.} For a finite group $G$ denote by $V$ its group algebra $k[G]$ with product $m$ and coproduct $\Delta$. Then, the map $T:=(id \otimes \Delta) \circ (m \otimes id): V\otimes V \to V\otimes V$ is a fusion operator.

\section{Results}\label{Results}

In this section, we combine the pentagon equation and the $A$ gate and the evolution operator of the 1D Heisenberg model to the compression of quantum circuits. Initially, we discuss how the pentagon equation is understood as a quantum circuit offering a ``local/non-local" form of duality.  Then, we show that the $A$ gate and the evolution operator of the 1D Heisenberg model for $N=2$ are both solutions of the pentagon equation subject to a number of constraining equations. 

\subsection*{Local/non-local duality}

% interactions between qubits adjacent in the coupling map
% circuit that involves only Heisenberg-type  local interactions 

Our first result is an observation of the nature of the solutions of the pentagon equation and the relevance they can have in the context of quantum circuits. 

Purely from the diagrammatic definition of the pentagon equation, Fig.~\ref{fig:fusion}, and especially the 3-cocycle condition, Fig.~\ref{fig:3co}, it is possible to realize that should 2-qubit fusion operators of interest be found, the corresponding circuits can be viewed in two dual ways. Before we proceed, we stress again that ``Heisenberg interactions" refer only to nearest-neighbor interactions, and we interchangeably use it with the phrase ``local interactions". Similarly, as mentioned in the introduction, non-nearest neighbor (long-range) interactions refer to (2-qubit interactions) non-adjacent qubits.

%at the cost of non-local interactions, one can reduce the number of Heisenberg-type interactions, effectively reducing the corresponding subcircuit depth. However, this requires the ability to enforce the corresponding entangling interactions $T$. Currently, for most quantum computing architectures with limited interqubit connectivity, this requires the application of SWAP gates. 

\comment{
Assume that a circuit
}

%%%%%% GIORGOS 27/9 %%%%%%%%%
%%%%%% GIORGOS 27/9 %%%%%%%%%
%%%%%% GIORGOS 27/9 %%%%%%%%%
Let us consider one point of view first:  one can reduce the number of non-Heisenberg-type interactions, effectively reducing the corresponding subcircuit depth. However, this requires the ability to enforce the corresponding entangling interactions $T$ which are solutions of the pentagon equation. Currently, for most quantum computing architectures with limited interqubit connectivity, implementation of non-local interactions requires the application of SWAP gates and as a result finding appropriate fusion operators $T$ could have a significant impact on circumventing SWAP gates whose high cost is due to the three CNOT that implement them. 

Interestingly, this point of view can be reversed. Assuming that interesting fusion operators are found, one can utilize the pentagon equation to expand a 
circuit that involves Heisenberg-type interactions to a circuit that also involves non-Heisenberg-type  local interactions at the cost of increased depth. 
%%%%%% GIORGOS 27/9 %%%%%%%%%
%%%%%% GIORGOS 27/9 %%%%%%%%%
%%%%%% GIORGOS 27/9 %%%%%%%%%

%Interestingly, this point of view can be reversed. Assuming that interesting fusion operators are found, one can utilize the pentagon equation to expand a circuit that involves non-Heisenberg-type  interactions to a circuit that involves only  Heisenberg-type  local interactions at the cost of increased depth. 

Although this may seem counterintuitive, with the current NISQ devices \cite{weidenfeller2022scaling,bravyi2022future}, where the effective depth of quantum circuits as well as the connectivity is limited and expensive, one could try to envision situations where highly complex circuits are easier to control using locally-entangling gates. Indeed, with substantial numbers of qubits, high connectivity is unlikely \cite{bravyi2022future}, considering the quadratic growth in the number of pairs of qubits, although there exists some evidence 
\cite{isenhower2011multibit,pagano2020quantum}
that direct application of non-local entangling interactions at the pulse level seems feasible, e.g., in cold-atomic architectures.

\subsection*{Transpiling Quantum Circuits}

Our second result is essentially positive. It shows how to apply the pentagon equation to transpiling quantum circuits. In particular, the pentagon equation makes it possible to transpile a certain quantum circuit with five gates into a quantum circuit with two gates in which the twist of the quantum wires is done using SWAP gates. Specifically:

\begin{theorem} 
\label{thm3}
Let $T$ be a unitary gate and a quantum circuit which contains the following component as a sub-circuit:
\begin{center}
\quad \begin{quantikz}
& \qw          & \qw                      &  \gate[2]{T}    &  \qw                         & \gate[2]{T}               & \qw & \ \\
& \gate[2]{T}  & \gate[2]{\textsc{SWAP}}  &                 &\gate[2]{\textsc{SWAP}}  &       & \qw &  \\
 &              &                          & \qw             &                              &   \qw                & \qw & 
\end{quantikz}
\end{center} If $T$ satisfies the pentagon equation \eqref{def:pen}, then the above component can be replaced by  \begin{center}
\begin{quantikz}
& \gate[2]{T}  &  \qw         &  \qw  & \\
&              & \gate[2]{T} &  \qw& \\
& \qw          &             &  \qw& 
 \end{quantikz}.\end{center} 
\end{theorem}

\begin{proof}
It is straightforward from the assumption that $T$ satisfies the pentagon equation \eqref{def:pen}.
\end{proof}

Theorem~\ref{thm3} provides a  way to transpile a quantum circuit assuming that the unitary gate $T$ is a fusion operator. The $A$-gate satisfies the pentagon equation and hence Theorem~\ref{thm3} can be applied for transipiling a quantum circuit for specific values of its coefficients of its coefficients $c_{1}$, $c_{2}$ and $c_{3}$. Precisely, we have the following.

\begin{proposition}
The $A$-gate, as defined in Eq. \eqref{Agate}, satisfies the pentagon equation \eqref{def:pen} if and only if the equations of Fig~\ref{fig:eqsA} are satisfied. 
\end{proposition}

\begin{proof}
Substituting the $A$ gate into the pentagon equation and equating both sides one arrives at solving the system of equations in Fig. \ref{fig:eqsA}. 
\end{proof}

Solving the equations of Fig.~\ref{fig:eqsA} one obtains $c_{1}=c_{2}=0$ and $c_{3}=-2\pi k$ where $k \in \mathbb{Z}$. This yields 
the triple eigenvalue of 1 in the corresponding unitary. The convex hull of the eigenvalues will thus not include the origin, which does not make it possible to reapply the proof technique of Corollary 2 in Ref. \cite{Vala_longPaper} to show that the parameter values yield a perfect entangler. Therefore, the $A$ gate is a fusion operator albeit of limited use. 

Based on the relation of the matrix $A$ and the evolution operator of the 1D Heisenberg model, cf. Remark (\ref{remark1}), we have the following.

\begin{corollary}
The evolution operator \eqref{def:evOp} of the 1D Heisenberg model for $N=2$ is a solution of the Pentagon equation \eqref{def:pen} if and only if the equations in Fig.~\ref{fig:eqsH} are satisfied. That is the case, when $\theta_{x}=\theta_{y}=0$ and $\theta_{z}=-k\pi$ for $k\in \mathbb{Z}$.
\end{corollary} 

\begin{proof}
Substituting the evolution operator of the 1D Heisenberg model into the pentagon equation and equating both sides, one arrives at solving the system of equations of Fig. \ref{fig:eqsH}. Solving them, we obtain the values for the coefficients $\theta_{x}$, $\theta_{y}$ and $\theta_{z}$.
\end{proof}

%\begin{align}
%e^{i \theta_{z} } \cos ^2(\gamma)&=e^{i \theta_{z} } \cos (\gamma) \left(2\cos ^2(\gamma)-1\right)\\
%e^{i \theta_{z} } \sin ^2(\gamma)&=2 e^{i \theta_{z} } \cos (\gamma) \sin ^2(\gamma)\\
%e^{-i \theta_{z} } \cos (\gamma) \cos (\zeta)&=e^{-i \theta_{z} } \cos (\zeta) \left(2\cos ^2(\gamma)-1\right)\\
%e^{-i \theta_{z} } \sin (\gamma) \sin (\zeta)&=2 e^{-i \theta_{z} } \cos (\gamma) \sin (\gamma) \sin (\zeta)\\
%i e^{i \theta_{z} } \cos (\gamma) \sin (\gamma)&=i e^{i \theta_{z} } \sin (\gamma) \left(2\cos ^2(\gamma)-1\right)\\
%i e^{i \theta_{z} } \cos (\gamma) \sin (\gamma)&=2 i e^{i \theta_{z} } \cos ^2(\gamma) \sin (\gamma)\\
%i e^{-i \theta_{z} } \cos (\gamma) \sin (\zeta)&=i e^{-i \theta_{z} } \left(2\cos ^2(\gamma)-1\right) \sin (\zeta)\\
%i e^{-\text{i$\theta_{z} $}} \cos (\zeta) \sin (\gamma)&=2 i e^{-\text{i$\theta_{z} $}} \cos (\gamma) \cos (\zeta) \sin (\gamma)\\
%i e^{-i \theta_{z} } \cos (\gamma) \sin (\gamma)&=i e^{-i \theta_{z} } \sin (\gamma) \left(2\cos ^2(\gamma)-1\right)
%\end{align}

\begin{figure*}[htb!]
\begin{widetext}
\begin{align*}\label{equ1}
e^{2 \im \theta_{z}} \cos ^2\left(\theta_{x}-\theta_{y}\right)-e^{\im\theta_{z}} \cos \left(\theta_{x}-\theta_{y}\right) \left(e^{2 \im\theta_{z}} \cos ^2\left(\theta_{x}-\theta_{y}\right)-e^{2 \im\theta_{z}} \sin ^2\left(\theta_{x}-\theta_{y}\right)\right)&=0\\
\im e^{2 \im\theta_{z}} \sin \left(\theta_{x}-\theta_{y}\right) \cos \left(\theta_{x}-\theta_{y}\right)-\im e^{\im \theta_{z}} \sin \left(\theta_{x}-\theta_{y}\right) \left(e^{2 \im \theta_{z}} \cos ^2\left(\theta_{x}-\theta_{y}\right)-e^{2 \im \theta_{z}} \sin ^2\left(\theta_{x}-\theta_{y}\right)\right)&=0\\
 \im e^{2 \im \theta_{z}} \sin \left(\theta_{x}-\theta_{y}\right) \cos \left(\theta_{x}-\theta_{y}\right)-2 \im e^{3 \im \theta_{z}} \sin \left(\theta_{x}-\theta_{y}\right) \cos ^2\left(\theta_{x}-\theta_{y}\right)&=0\\
 2 e^{3 \im \theta_{z}} \sin ^2\left(\theta_{x}-\theta_{y}\right) \cos \left(\theta_{x}-\theta_{y}\right)-e^{2 \im \theta_{z}} \sin ^2\left(\theta_{x}-\theta_{y}\right)&=0\\
\cos \left(\theta_{x}-\theta_{y}\right) \cos \left(\theta_{x}+\theta_{y}\right)-e^{-\im \theta_{z}} \cos \left(\theta_{x}+\theta_{y}\right) \left(e^{2 \im \theta_{z}} \cos ^2\left(\theta_{x}-\theta_{y}\right)-e^{2 \im \theta_{z}} \sin ^2\left(\theta_{x}-\theta_{y}\right)\right)&=0\\
\im \sin \left(\theta_{x}+\theta_{y}\right) \cos \left(\theta_{x}-\theta_{y}\right)-\im e^{-\im \theta_{z}} \sin \left(\theta_{x}+\theta_{y}\right) \left(e^{2 \im \theta_{z}} \cos ^2\left(\theta_{x}-\theta_{y}\right)-e^{2 \im \theta_{z}} \sin ^2\left(\theta_{x}-\theta_{y}\right)\right)&=0\\ 
\im \sin \left(\theta_{x}-\theta_{y}\right) \cos \left(\theta_{x}+\theta_{y}\right)-2 \im e^{\im \theta_{z}} \sin \left(\theta_{x}-\theta_{y}\right) \cos \left(\theta_{x}-\theta_{y}\right) \cos \left(\theta_{x}+\theta_{y}\right)&=0\\
-\sin \left(\theta_{x}-\theta_{y}\right) \sin \left(\theta_{x}+\theta_{y}\right)+2 e^{\im \theta_{z}} \sin \left(\theta_{x}-\theta_{y}\right) \sin \left(\theta_{x}+\theta_{y}\right) \cos \left(\theta_{x}-\theta_{y}\right)&=0\\ 
\cos \left(\theta_{x}-\theta_{y}\right) \cos \left(\theta_{x}+\theta_{y}\right)-e^{\im \theta_{z}} \cos \left(\theta_{x}-\theta_{y}\right) \left(e^{-2 \im \theta_{z}} \cos ^2\left(\theta_{x}+\theta_{y}\right)-e^{-2 \im \theta_{z}} \sin ^2\left(\theta_{x}+\theta_{y}\right)\right)&=0\\ 
\im \sin \left(\theta_{x}-\theta_{y}\right) \cos \left(\theta_{x}+\theta_{y}\right)-\im e^{\im \theta_{z}} \sin \left(\theta_{x}-\theta_{y}\right) \left(e^{-2 \im \theta_{z}} \cos ^2\left(\theta_{x}+\theta_{y}\right)-e^{-2 \im \theta_{z}} \sin ^2\left(\theta_{x}+\theta_{y}\right)\right)&=0\\  
\im \sin \left(\theta_{x}+\theta_{y}\right) \cos \left(\theta_{x}-\theta_{y}\right)-2 \im e^{-\im \theta_{z}} \sin \left(\theta_{x}+\theta_{y}\right) \cos \left(\theta_{x}-\theta_{y}\right) \cos \left(\theta_{x}+\theta_{y}\right)&=0\\  
-\sin \left(\theta_{x}-\theta_{y}\right) \sin \left(\theta_{x}+\theta_{y}\right)+2 e^{-\im \theta_{z}} \sin \left(\theta_{x}-\theta_{y}\right) \sin \left(\theta_{x}+\theta_{y}\right) \cos \left(\theta_{x}+\theta_{y}\right)&=0\\
e^{-2 \im \theta_{z}} \cos ^2\left(\theta_{x}+\theta_{y}\right)-e^{-\im \theta_{z}} \cos \left(\theta_{x}+\theta_{y}\right) \left(e^{-2 \im \theta_{z}} \cos ^2\left(\theta_{x}+\theta_{y}\right)-e^{-2 \im \theta_{z}} \sin ^2\left(\theta_{x}+\theta_{y}\right)\right)&=0\\
\im e^{-2 \im \theta_{z}} \sin \left(\theta_{x}+\theta_{y}\right) \cos \left(\theta_{x}+\theta_{y}\right)-\im e^{-\im \theta_{z}} \sin \left(\theta_{x}+\theta_{y}\right) \left(e^{-2 \im \theta_{z}} \cos ^2\left(\theta_{x}+\theta_{y}\right)-e^{-2 \im \theta_{z}} \sin ^2\left(\theta_{x}+\theta_{y}\right)\right)&=0\\ 
\im e^{-2 \im \theta_{z}} \sin \left(\theta_{x}+\theta_{y}\right) \cos \left(\theta_{x}+\theta_{y}\right)-2 \im e^{-3 \im \theta_{z}} \sin \left(\theta_{x}+\theta_{y}\right) \cos ^2\left(\theta_{x}+\theta_{y}\right)&=0\\  
2 e^{-3 \im \theta_{z}} \sin ^2\left(\theta_{x}+\theta_{y}\right) \cos \left(\theta_{x}+\theta_{y}\right)-e^{-2 \im \theta_{z}} \sin ^2\left(c_1+\theta_{y}\right)&=0
\end{align*}
\caption{The equations to be satisfied by the 1D
Heisenberg model such as the Hamiltonian of the model corresponds to the fusion operator $T$ of Fig. \ref{fig:fusion}. Recall that satisfying these equations is an absolute requirement in order to achieve the compression of Fig. \ref{fig:fusion}. This set of equations resemble  Ref. \cite[Eqs. (14)-(29)]{Gulania2021} in the context of solutions of the YBE.} 
\label{fig:eqsH}
\end{widetext}
\end{figure*}

\begin{figure*}[tbh!]
\begin{widetext}
\begin{align*}
e^{\im c_3} \cos ^2\left(\tfrac{1}{2} \left(c_1-c_2\right)\right)-e^{\tfrac{\im c_3}{2}} \cos \left(\tfrac{1}{2} \left(c_1-c_2\right)\right) \left(e^{\im c_3} \cos ^2\left(\tfrac{1}{2} \left(c_1-c_2\right)\right)-e^{\im c_3} \sin ^2\left(\tfrac{1}{2} \left(c_1-c_2\right)\right)\right)&=0 \\
\im e^{\im c_3} \sin \left(\tfrac{1}{2} \left(c_1-c_2\right)\right) \cos \left(\tfrac{1}{2} \left(c_1-c_2\right)\right)-\im e^{\tfrac{\im c_3}{2}} \sin \left(\tfrac{1}{2} \left(c_1-c_2\right)\right) \left(e^{\im c_3} \cos ^2\left(\tfrac{1}{2} \left(c_1-c_2\right)\right)-e^{\im c_3} \sin ^2\left(\tfrac{1}{2} \left(c_1-c_2\right)\right)\right)&=0 \\
\im e^{\im c_3} \sin \left(\tfrac{1}{2} \left(c_1-c_2\right)\right) \cos \left(\tfrac{1}{2} \left(c_1-c_2\right)\right)-2 \im e^{\tfrac{3 \im c_3}{2}} \sin \left(\tfrac{1}{2} \left(c_1-c_2\right)\right) \cos ^2\left(\tfrac{1}{2} \left(c_1-c_2\right)\right)&=0 \\
2 e^{\tfrac{3 \im c_3}{2}} \sin ^2\left(\tfrac{1}{2} \left(c_1-c_2\right)\right) \cos \left(\tfrac{1}{2} \left(c_1-c_2\right)\right)-e^{\im c_3} \sin ^2\left(\tfrac{1}{2} \left(c_1-c_2\right)\right)&=0 \\
\cos \left(\tfrac{1}{2} \left(c_1-c_2\right)\right) \cos \left(\tfrac{1}{2} \left(c_1+c_2\right)\right)-e^{-\tfrac{1}{2} \left(\im c_3\right)} \cos \left(\tfrac{1}{2} \left(c_1+c_2\right)\right) \left(e^{\im c_3} \cos ^2\left(\tfrac{1}{2} \left(c_1-c_2\right)\right)-e^{\im c_3} \sin ^2\left(\tfrac{1}{2} \left(c_1-c_2\right)\right)\right)&=0 \\
\im \sin \left(\tfrac{1}{2} \left(c_1+c_2\right)\right) \cos \left(\tfrac{1}{2} \left(c_1-c_2\right)\right)-\im e^{-\tfrac{1}{2} \left(\im c_3\right)} \sin \left(\tfrac{1}{2} \left(c_1+c_2\right)\right) \left(e^{\im c_3} \cos ^2\left(\tfrac{1}{2} \left(c_1-c_2\right)\right)-e^{\im c_3} \sin ^2\left(\tfrac{1}{2} \left(c_1-c_2\right)\right)\right)&=0 \\
\im \sin \left(\tfrac{1}{2} \left(c_1-c_2\right)\right) \cos \left(\tfrac{1}{2} \left(c_1+c_2\right)\right)-2 \im e^{\tfrac{\im c_3}{2}} \sin \left(\tfrac{1}{2} \left(c_1-c_2\right)\right) \cos \left(\tfrac{1}{2} \left(c_1-c_2\right)\right) \cos \left(\tfrac{1}{2} \left(c_1+c_2\right)\right)&=0 \\
-\sin \left(\tfrac{1}{2} \left(c_1-c_2\right)\right) \sin \left(\tfrac{1}{2} \left(c_1+c_2\right)\right)+2 e^{\tfrac{\im c_3}{2}} \sin \left(\tfrac{1}{2} \left(c_1-c_2\right)\right) \sin \left(\tfrac{1}{2} \left(c_1+c_2\right)\right) \cos \left(\tfrac{1}{2} \left(c_1-c_2\right)\right)&=0 \\
\cos \left(\tfrac{1}{2} \left(c_1-c_2\right)\right) \cos \left(\tfrac{1}{2} \left(c_1+c_2\right)\right)-e^{\tfrac{\im c_3}{2}} \cos \left(\tfrac{1}{2} \left(c_1-c_2\right)\right) \left(e^{-\im c_3} \cos ^2\left(\tfrac{1}{2} \left(c_1+c_2\right)\right)-e^{-\im c_3} \sin ^2\left(\tfrac{1}{2} \left(c_1+c_2\right)\right)\right)&=0 \\
\im \sin \left(\tfrac{1}{2} \left(c_1-c_2\right)\right) \cos \left(\tfrac{1}{2} \left(c_1+c_2\right)\right)-\im e^{\tfrac{\im c_3}{2}} \sin \left(\tfrac{1}{2} \left(c_1-c_2\right)\right) \left(e^{-\im c_3} \cos ^2\left(\tfrac{1}{2} \left(c_1+c_2\right)\right)-e^{-\im c_3} \sin ^2\left(\tfrac{1}{2} \left(c_1+c_2\right)\right)\right)&=0 \\
\im \sin \left(\tfrac{1}{2} \left(c_1+c_2\right)\right) \cos \left(\tfrac{1}{2} \left(c_1-c_2\right)\right)-2 \im e^{-\tfrac{1}{2} \left(\im c_3\right)} \sin \left(\tfrac{1}{2} \left(c_1+c_2\right)\right) \cos \left(\tfrac{1}{2} \left(c_1-c_2\right)\right) \cos \left(\tfrac{1}{2} \left(c_1+c_2\right)\right)&=0 \\
-\sin \left(\tfrac{1}{2} \left(c_1-c_2\right)\right) \sin \left(\tfrac{1}{2} \left(c_1+c_2\right)\right)+2 e^{-\tfrac{1}{2} \left(\im c_3\right)} \sin \left(\tfrac{1}{2} \left(c_1-c_2\right)\right) \sin \left(\tfrac{1}{2} \left(c_1+c_2\right)\right) \cos \left(\tfrac{1}{2} \left(c_1+c_2\right)\right)&=0 \\
e^{-\im c_3} \cos ^2\left(\tfrac{1}{2} \left(c_1+c_2\right)\right)-e^{-\tfrac{1}{2} \left(\im c_3\right)} \cos \left(\tfrac{1}{2} \left(c_1+c_2\right)\right) \left(e^{-\im c_3} \cos ^2\left(\tfrac{1}{2} \left(c_1+c_2\right)\right)-e^{-\im c_3} \sin ^2\left(\tfrac{1}{2} \left(c_1+c_2\right)\right)\right)&=0 \\
\im e^{-\im c_3} \sin \left(\tfrac{1}{2} \left(c_1+c_2\right)\right) \cos \left(\tfrac{1}{2} \left(c_1+c_2\right)\right)-\im e^{-\tfrac{1}{2} \left(\im c_3\right)} \sin \left(\tfrac{1}{2} \left(c_1+c_2\right)\right) \left(e^{-\im c_3} \cos ^2\left(\tfrac{1}{2} \left(c_1+c_2\right)\right)-e^{-\im c_3} \sin ^2\left(\tfrac{1}{2} \left(c_1+c_2\right)\right)\right)&=0 \\
\im e^{-\im c_3} \sin \left(\tfrac{1}{2} \left(c_1+c_2\right)\right) \cos \left(\tfrac{1}{2} \left(c_1+c_2\right)\right)-2 \im e^{-\tfrac{1}{2} \left(3 \im c_3\right)} \sin \left(\tfrac{1}{2} \left(c_1+c_2\right)\right) \cos ^2\left(\tfrac{1}{2} \left(c_1+c_2\right)\right)&=0 \\
2 e^{-\tfrac{1}{2} \left(3 \im c_3\right)} \sin ^2\left(\tfrac{1}{2} \left(c_1+c_2\right)\right) \cos \left(\tfrac{1}{2} \left(c_1+c_2\right)\right)-e^{-\im c_3} \sin ^2\left(\tfrac{1}{2} \left(c_1+c_2\right)\right)&=0
\end{align*}
\caption{The equations to be satisfied by the $A$ gate of the fusion operator of Fig. \ref{fig:fusion}.} 
\label{fig:eqsA}
\end{widetext}
\end{figure*}

%\section{Application: compression of quantum circuits} \label{App}

%Compressing the number of gates is a challenge in the theory of quantum circuits, see \cite{Gulania2021} for a recent approach close to the sense of this note. In addition to this, implementing a swap gate is costly even if can be replaced by 3 CNOT gates.

%\begin{quantikz}
% & \gate[2]{T} & \qw  \\
 %& \qw & \qw & \gate[3]{T} & \qw  \\
 %& \gate[2]{T} & \qw  \\
 %& \qw &\qw \\
%\end{quantikz}

% CA: the following pic is incorrect
%\begin{figure}
  %  \centering
  %  \includegraphics[scale=0.2]{IMG_0231.jpg}
 %   \caption{Caption}
  %  \label{fig:my_label}
%\end{figure}

%\section{Simulation}
%Jakub's suggestions are the following.

%\begin{enumerate}
%    \item Initially find all the possible gates that satisfy the pentagon, standard or not.
%    \item Find some commonly used circuit, where the reduction helps a lot. No need to do a super detailed analysis like \cite{https://doi.org/10.48550/arxiv.2202.14025}
    
%    \item Jakub mentioned  (Exact and Practical) Pattern Matching for Quantum circuits, referring to \cite{https://doi.org/10.48550/arxiv.1909.05270}. 
%\end{enumerate}

\section{Summary and Conclusion}

In this note, we investigate under which conditions the $A$ gate and the unitary $e^{\im \hat{H}}$, where $\hat{H}$ is the Hamiltonian of the 1D Heisenberg model, satisfy the pentagon equation \eqref{def:pen} in the context of circuit compression. By carefully analyzing the pentagon equation, we were able to show that a quantum circuit composed of 2-qubit interactions that include non-local interactions can be compressed to a local interactions only quantum circuit if the corresponding gates are fusion operators, that is, for a particular set of parameters that satisfy the constraining equations. 
%We plan to examine it in a future paper.

% for example, both in neutral atom architectures (regular square lattices) and superconducting qubits
It is natural to wonder if it is possible to implement non-local, at least next-to-nearest-neighbor, interactions in the current NISQ architectures. As implied in several occasions in the preceding content, there are indeed cases where this is possible, especially in cold-atomic
% Rydberg is for neutral atoms
architectures based on the Rydberg blockade \cite{isenhower2011multibit} or trapped ions \cite{pagano2020quantum}, although it is not straightforward to see how this would scale to more than a few such interactions since it is highly non-trivial to find and control the optimal pulses that would implement such interactions while eliminating possible cross-talk effects. (This challenge is often referred to as ``crowding".) In that sense finding fusion operators that could perform the reduction as described previously would be highly desired. Furthemore, in future iterations of NISQ devices and early fault-tolerant (EFT) devices, there is potential for accurate implementation of next-to-nearest-neighbor interactions required for demanding quantum algorithms, and we hope that finding interesting fusion operators and other such local/non-local dualities could have a substantial impact by reducing the corresponding number of SWAP gates when direct non-nearest neighbor interactions are not possible.

Generally, finding solutions that satisfy the pentagon equation, fusion operators, and transpiling quantum circuits in the approach presented above, is highly non-trivial. While we only considered the pentagon equation in the context of the $A$ gate and the evolution operator of the 1D Heisenberg model, coming up with the ``the least trivial solution," one can argue that more solutions exist and can actually perform circuit reduction. Such solutions can potentially arise from other integrable systems.

Another approach to finding fusion operators would be to follow the mathematical recipe given in Remark \ref{rem:recipe} for  the group ring let $\mathbb{C}[U(4)]$ of the Lie group $U(4)$. To be more precise, $\mathbb{C}[U(4)]$ is the set of all linear combinations of finitely many elements of $U(4)$ with coefficients in $\mathbb{C}$ and has a bialgebra structure. The product and the coproduct of $\mathbb{C}[U(4)]$ define a fusion operator. We plan to expand this direction in a future note.

\subsubsection*{{\bf Aknowledgements}}
We would like to thank Philip Intallura and Juan Miguel Nieto Garc\'ia for usefull discussions, comments and suggestions on an early version of the draft. 
J.M. acknowledges the support of the OP
VVV project CZ.02.1.01/0.0/0.0/16\_019/0000765 ``Research Center for Informatics".

\subsubsection*{{\bf Disclaimer}}
This paper was prepared for information purposes
and is not a product of HSBC Bank Plc. or its affiliates.
Neither HSBC Bank Plc. nor any of its affiliates make
any explicit or implied representation or warranty and
none of them accept any liability in connection with
this paper, including, but limited to, the completeness,
accuracy, reliability of information contained herein and
the potential legal, compliance, tax, or accounting effects
thereof. This document is not intended as investment
research or investment advice, or a recommendation,
offer or solicitation for the purchase or sale of any security, financial instrument, financial product, or service,
or to be used in any way for evaluating the merits of
participating in any transaction.

\bibliography{apssamp}% Produces the bibliography via BibTeX.

\end{document}